\documentclass[a4paper]{article}
\usepackage[a4paper,includeheadfoot,margin=2.54cm]{geometry}

\usepackage[T1]{fontenc}
\usepackage{amsthm}
\usepackage{amsmath}
\usepackage{amssymb}
\usepackage[affil-it]{authblk}
\usepackage{microtype}
\usepackage{cite}
\usepackage{soul}
\usepackage{xspace}

\usepackage{color}

\usepackage{etoolbox}
\usepackage{lmodern}
\makeatletter
\patchcmd{\@maketitle}{\LARGE \@title}{\fontsize{16}{19.2}\selectfont\@title}{}{}\makeatother

\usepackage{graphicx}
\graphicspath{{./figures/}}

\theoremstyle{plain}
\newtheorem{theorem}{Theorem}[section]
\newtheorem{lemma}[theorem]{Lemma}
\newtheorem{observation}[theorem]{Observation}

\title{Routing in Polygonal Domains\footnote{
A preliminary version appeared as
B.~Banyassady, M-K.~Chiu, M.~Korman, W.~Mulzer, A.~v.~Renssen,
M.~Roeloffzen, P.~Seiferth, Y.~Stein, B.~Vogtenhuber,
and M.~Willert. \emph{Routing in Polygonal Domains.} Proc.~28th ISAAC,
pp.~10:1--10.13.
   BB was supported in part by DFG project MU/3501-2.
   MC, AvR and MR were supported by JST ERATO Grant Number JPMJER1201, Japan.  
   MK was supported in part by KAKENHI Nos. 15H02665 and 17K12635, Japan.
   WM was supported in part by ERC StG 757609.
   PS was supported in part by DFG project MU/3501-1.
   YS was supported by the DFG within the research training group 
   `Methods for Discrete Structures' (GRK 1408) and by GIF 
   grant 1161.}}    

\author[1]{Bahareh Banyassady}
\author[2]{Man-Kwun Chiu}
\author[3]{Matias Korman}
\author[1]{Wolfgang Mulzer}
\author[4]{\\Andr\'{e} van Renssen}
\author[5]{Marcel Roeloffzen}
\author[1]{Paul Seiferth}
\author[1]{\\Yannik Stein}
\author[6]{Birgit Vogtenhuber}
\author[1]{Max Willert}

\affil[1]{Institut f\"ur Informatik, Freie Universit\"at Berlin, Germany\\
\texttt{\{bahareh,mulzer,pseiferth,yannikstein,willerma\}@inf.fu-berlin.de}}
\affil[2]{Institut f\"ur Informatik, Freie Universit\"at Berlin, Germany\\
  \texttt{cmk.kenny@gmail.com}}
\affil[3]{Tohoku University, Sendai, Japan\\
  \texttt{mati@dais.is.tohoku.ac.jp}}
\affil[4]{The University of Sydney, Sydney, Australia\\
  \texttt{andre.vanrenssen@sydney.edu.au}}
\affil[5]{Department of Mathematics and Computer Science, TU Eindhoven, 
Eindhoven, the Netherlands\\
\texttt{m.j.m.roeloffzen@tue.nl}}
\affil[6]{Institute of Software Technology, Graz University of 
Technology, Graz, Austria\\
\texttt{bvogt@ist.tugraz.at}}

\date{}
\begin{document}
\maketitle

\newcommand{\mati}[1]{\niceremark{Matias}{#1}{red}}
\newcommand{\kenny}[1]{\niceremark{Kenny}{#1}{red}}
\newcommand{\maxe}[1]{\niceremark{Max}{#1}{blue}}
\newcommand{\wolfgang}[1]{\niceremark{Wolfgang}{#1}{red}}
\newcommand{\rewritten}[1]{{\color{red}#1}}
\newcommand{\niceremark}[3]{\textcolor{#3}{\textsc{#1:} \marrow\textsf{#2}}}
\newcommand{\marrow}{\marginpar[\hfill$\longrightarrow$]{$\longleftarrow$}}

\newcommand{\abs}[1]{\vert\overline{#1}\vert}
\newcommand{\Oe}{O}
\newcommand{\C}{\mathcal{C}}
\newcommand{\eps}{\varepsilon}
\newcommand{\VG}{\operatorname{VG}}
\newcommand{\innt}{\operatorname{int}}
\newcommand{\vis}{\operatorname{vis}}
\newcommand{\outdeg}{\operatorname{outdeg}}
\newcommand{\R}{\mathbb{R}}

\newcommand{\etal}{\emph{et al.}\xspace}

\begin{abstract}
We consider the problem of routing a data packet through
the visibility graph of a polygonal domain $P$ with $n$
vertices and $h$ holes.  We may preprocess $P$ to obtain 
a \emph{label} and a \emph{routing table} for each vertex of $P$. 
Then, we must be able to route a data packet between
any two vertices $p$ and $q$ of $P$, where each step 
must use only the label of the target node $q$
and the routing table of the current node. 

For any fixed $\varepsilon > 0$, 
we present a routing scheme that always achieves
a routing path whose length exceeds the shortest path by
a factor of at most $1 + \varepsilon$.
The labels have $O(\log n)$ bits, and the 
routing tables are of size $O((\varepsilon^{-1}+h)\log n)$. 
The preprocessing time is $O(n^2\log n)$.
It can be improved to $O(n^2)$ for 
simple polygons.
\end{abstract} 

\section{Introduction}
Routing is a crucial problem in distributed graph 
algorithms~\cite{GiordanoSt04,PelegUp89}. 
We would like to preprocess a given graph $G$ in 
order to support the following task: given a data packet
that lies at some \emph{source} node $p$ of $G$, route 
the packet to a given \emph{target} node $q$ in $G$ 
that is identified by its \emph{label}. We expect three 
properties from our routing scheme: first, it should be 
\emph{local}, i.e., in order to determine the next step 
for the packet, it should use only information 
stored with the current node of $G$ or with the packet 
itself. Second, the routing scheme should be \emph{efficient}, 
meaning that the packet should not travel much
more than the shortest path distance between $p$ and $q$. 
The ratio between the length of the routing path and the shortest 
path in the graph is also called \emph{stretch factor}. 
Third, it should be \emph{compact}: the total space requirement 
should be as small as possible.

Here is an obvious solution: for 
each node $v$ of $G$, we store at $v$ the complete shortest 
path tree for $v$. Thus, given the
label of a target node $q$, we can send the packet
for one more step along the shortest path from $v$ to $q$.
Then, the routing scheme will have perfect
efficiency, sending each packet along a shortest
path. However, this method requires that each node 
stores its entire shortest path tree, making it
not compact. Thus, the challenge 
lies in finding the right balance between
the conflicting goals of compactness and efficiency.

Thorup and Zwick introduced the notion of a 
\emph{distance oracle}~\cite{ThorupZw05}. 
Given a graph $G$, the goal is to construct a compact 
data structure to quickly answer \emph{distance queries} 
for any two nodes in $G$. A
routing scheme can be seen as a distributed implementation of a
distance oracle~\cite{RodittyTo16}. 

The problem of constructing a compact routing
scheme for a general graph has been studied 
for a long time~\cite{AbrahamGa11,AwerbuchBNLiPe90,Chechik13,Cowen01,
EilamGaPe03,RodittyTo15,RodittyTo16}. One of the most recent results,
by Roditty and Tov, dates from 2016~\cite{RodittyTo16}. 
They developed a routing scheme for a general graph $G$ 
with $n$ vertices and $m$ edges. Their scheme needs to store a 
poly-logarithmic number of bits with the packet, and
it routes a message from $p$ to $q$ on a path with length 
$\Oe(k\Delta+m^{1/k})$, where $\Delta$ is the 
shortest path distance between $p$ and $q$ and $k > 2$ is any 
fixed integer. 
The routing tables use $mn^{\Oe(1/\sqrt{\log n})}$ total space. 
In general graphs, any routing scheme with constant stretch factor
needs to store $\Omega(n^c)$ bits per node, for some constant 
$c>0$~\cite{PelegUp89}. Thus, it is natural to ask whether 
there are better algorithms for specialized graph classes. For 
instance, trees admit routing schemes that always 
follow the shortest path and that store $\Oe(\log n)$ bits at 
each node~\cite{FraigniaudGa01,SantoroKh85,ThorupZw01}. 
Moreover, in planar graphs,
for any fixed $\eps > 0$, there is a routing scheme with a 
poly-logarithmic number of bits in each routing table that 
always finds a path that is within a factor of 
$1 + \eps$ from optimal~\cite{Thorup04}.
Similar results are also available for
unit disk graphs~\cite{KaplanMuRoSe18}, and for metric spaces with 
bounded doubling dimension~\cite{KonjevodRiXi16}.

Another approach is called \emph{geometric routing}.
Here, the graph is embedded in a geometric space, and 
the routing algorithm has to determine the next vertex for 
the data packet based on the location of the source and the target 
vertex, the current vertex, and its neighbourhood, see 
for instance~\cite{BoseFavReVe17,BoseFavReVe15} and the references 
therein.  The most notable difference between geometric routing and 
our setting is that in geometric routing, vertices are generally 
not allowed to store routing tables, so that routing decisions 
are based solely on the geometric information available at the 
current vertex (and possibly information stored in the message). 
We note that the location of the source vertex may or may not be 
needed, depending on the routing algorithm. For example, the routing 
algorithm for triangulations by Bose and Morin~\cite{BoseMo04} uses 
the line segment between the source and the target for its 
routing decisions.  A recent result by 
Bose~\etal~\cite{BoseFavReVe17} is very close to our setting.
They show that when vertices do not store any routing tables 
(i.e., each vertex stores only the edges that can be followed 
from it), no geometric routing scheme can achieve stretch 
factor $o(\sqrt{n})$. This lower bound applies regardless of the 
amount of information that may be stores in the message. 

Here, we consider the class of visibility graphs of a 
polygonal domain. Let $P$ be such a polygonal domain with $h$ 
holes and $n$ vertices. 
Two vertices $p$ and $q$ in $P$ are connected by an edge if and only 
if they can \emph{see} each other, i.e., if and only if 
the line segment between $p$ and $q$ is contained in the 
(closed) region $P$. We note that this definition implies that the 
visibility graph contains the shortest path between any two vertices 
of the polygonal domain. The problem of computing a shortest 
path between two vertices in a polygonal domain has 
been well-studied in computational geometry~\cite{AsanoAsGuHeIm86,
BYCh94,GuibasHeLeShTa87,HershbergerSu99,KapoorMaMi97,KapoorMa88,
Mitchell91,Mitchell96,OvermarsWe88,SharirSc86,StorerRe94,Welzl85}. 
Nevertheless, to the best of our knowledge, prior to our work there 
have been no routing schemes for visibility graphs of polygonal 
domains that fall into our model.  

When we relax the requirement on the length of the path, we enter 
the domain of spanners: given a graph $G$, a subgraph $H$ of $G$ is 
a \emph{$k$-spanner} of $G$ if for all pairs of vertices $p$ and $q$
in $G$, $d_H(p, q) \leq k \cdot d_G(p, q)$, for $k \geq 1$. The 
spanning properties of various geometric graphs have been studied
extensively in the literature (see~\cite{BoseSm13,NarasimhanSm07} 
for a comprehensive overview). We briefly mention the results that 
are most closely related to the approach we will take here, namely 
Yao-graphs~\cite{Yao82} and 
$\Theta$-graphs~\cite{Clarkson87}. Intuitively, these graphs form 
geometric networks where each vertex connects to its nearest visible
vertex in a certain number of different directions (a formal 
definition is given in Section~\ref{sec:cones}). Both types of 
graphs are spanners, where the stretch factor depends on the number 
of cones used~\cite{BarbaBoCaReVe13,BarbaBoDaFaKeORoTaVeiX15,
BoseCaMoReSe16,BoseDaDoORSeSmWu12,BoseMoReVe15,DamianRa12,DamianNe17}. 
These graphs have also been considered for geometric routing 
purposes. For example, Bose~\etal~\cite{BoseFavReVe15} gave an 
optimal geometric routing algorithm for the half-$\Theta_6$-graph 
(the $\Theta$-graph with six cones where edges are added in every 
other cone). When considering obstacles, $\Theta$-graphs have 
recently been used to route on (subgraphs of) the visibility 
graph~\cite{BoseFavReVe17,BoseKovReVe17b,BoseKovReVe17}, though 
these algorithms do not provide a bound on the total length of the 
routing path, only on the number of edges followed by the routing 
scheme. However, as mentioned earlier, these geometric routing 
schemes cannot achieve a stretch factor of $o(\sqrt{n})$, as they 
are not allowed to store routing tables at the vertices. 

We introduce a routing scheme that, for any $\eps > 0$, needs 
$\Oe((1/\eps + h)\log n)$ bits in each routing table, and 
for any two vertices $p$ and $q$, it
produces a routing path that is within a factor of $1 + \eps$ of 
optimal.  This shows that by allowing a routing table 
at each vertex, we can do much better than in traditional
geometric routing, achieving a stretch factor that is arbitrarily close
to $1$.


\section{Preliminaries}

Let $G = (V, E)$ be an \emph{undirected}, \emph{connected} and 
\emph{simple} graph. In our model, $G$ is embedded in the Euclidean 
plane: a \emph{node} $p = (p_x,p_y)\in V$ corresponds to a point 
in the plane, and an edge $\{p,q\}\in E$ is represented by the 
line segment $\overline{pq}$. The \emph{length} $\abs{pq}$ of an
edge $\{p,q\}$ is the Euclidean distance between the 
points $p$ and $q$. The length of a shortest path between two 
nodes $p,q\in V$ is denoted by $d(p,q)$.

We formally define a \emph{routing scheme} for $G$. Each 
node $p$ of $G$ is assigned a \emph{label} 
$\ell(p) \in \{0,1\}^*$ that identifies it in the network. 
Furthermore, we store with $p$ a \emph{routing table}
$\rho(p) \in \{0, 1\}^*$. The routing scheme works as follows: 
the packet contains the label $\ell(q)$ of the target node $q$,
and initially it is situated at the start node $p$.
In each step of the routing algorithm, the packet resides at 
a current node $p' \in V$. It may consult the routing table
$\rho(p')$ of $p'$ and the label $\ell(q)$ of the target
to determine the next node $q'$ to which the packet is
forwarded. The node $q'$ must be a neighbor of $p'$ in $G$. 
This is repeated until the packet reaches its destination $q$. 
The scheme is modeled by a 
\emph{routing function} $f: \rho(V) \times \ell(V) \rightarrow V$. 

In the literature, there are varying definitions for 
the notion of a routing scheme~\cite{KaplanMuRoSe18, yan2012compact, 
RodittyTo16}. For example, we may sometimes store 
additional information in the \emph{header} of a data packet 
(it travels with the packet and can store information from past vertices).
Similarly, the 
routing function sometimes allows the use of an 
\emph{intermediate} target label. This is helpful for 
recursive routing schemes. Here, however, we will not need 
any of these additional capabilities.

As mentioned, the routing scheme operates by repeatedly
applying the routing function. More precisely, given a
start node $p \in V$ and a target label $\ell(q)$,
the scheme produces the sequence of nodes
$p_0 = p$ and $p_i = f(\rho(p_{i-1}), \ell(q))$, for $i \geq 1$. 
Naturally, we want routing schemes for which every packet reaches 
its desired destination. More precisely, a routing scheme is 
\emph{correct} if for any $p, q\in V$, there exists a finite 
$k = k(p,q) \geq 0$ such that $p_k = q$ (and $p_{i} \neq q$ for 
$0 \leq i <k $). We call $p_0, p_1, \dots,p_k$ the 
\emph{routing path} between $p$ and $q$. The \emph{routing distance} 
between $p$ and $q$ is defined as 
$d_{\rho}(p,q) = \sum_{i = 1}^k \abs{p_{i-1}p_i}$.

The quality of the routing scheme is measured by several parameters:
\begin{enumerate}
\item  the \emph{label size} 
  $\max_{p \in V} \vert \ell(p)\vert$,
\item the \emph{table size} 
  $\max_{p \in V} \vert\rho(p)\vert$,
\item the \emph{stretch factor}
  $\max_{p\neq q\in V}d_{\rho}(p,q)/d(p,q)$, and
\item the preprocessing time.
\end{enumerate}

Let $P$ be a polygonal domain with $n$ vertices. 
The \emph{boundary} $\partial P$ of $P$ consists of $h$ 
pairwise disjoint simple closed polygonal chains: one 
\emph{outer boundary} and $h-1$ \emph{hole boundaries}, or 
$h$ \emph{hole boundaries} with no outer boundary. All hole
boundaries lie inside the outer boundary, and no hole boundary lies
inside another hole boundary. In both cases, we say that $P$ has 
$h$ holes. The interior induced by any hole boundary and the exterior 
of the outer boundary are not contained in $P$. We denote 
the (open) \emph{interior} of $P$ by $\innt P$, i.e., 
$\innt P = P\setminus\partial P$. We assume that $P$ is in 
general position: no three vertices of $P$ lie on a common line, and
for each pair of vertices in $P$, the shortest path between them 
is unique. Let $n_i$, $0 \leq i \leq h-1$, be 
the number of vertices on the $i$-th boundary of $P$.
For each boundary $i$, we number the vertices from
$0$ to $n_i - 1$, in clockwise order if $i$ is a hole boundary, or 
in counterclockwise order if $i$ is the outer boundary. The $k$th 
vertex of the $i$th boundary is denoted by $p_{i,k}$.

Two points $p$ and $q$ in $P$ can \emph{see each other} in $P$ 
if and only if $\overline{pq} \subset P$. 
By our general position assumption, $\overline{pq}$
touches $\partial P$ only if $\overline{pq}$ is itself an
edge of $P$.
The \emph{visibility 
graph} of $P$, $\VG(P)$, has the same vertices
as $P$ and an edge between two vertices 
if and only if they see each other in $P$.
We show the following main theorem:

\begin{theorem}
\label{mainThm}
Let $P$ be a polygonal domain with $n$ vertices and $h$ holes. 
For any $\eps > 0$, we can construct a routing scheme for 
$\VG(P)$ with labels of $\Oe(\log n)$ bits and routing tables of 
$\Oe((1/\eps + h)\log n)$ bits per vertex. For any two sites 
$p, q\in P$, the scheme produces a routing path with stretch 
factor at most $1 + \eps$. The preprocessing time is 
$\Oe(n^2\log n)$. If $P$ is a simple polygon, the preprocessing
time reduces to $O(n^2)$.
\end{theorem}


\section{Cones in Polygonal Domains}
\label{sec:cones}

Let $P$ be a polygonal domain with $n$ vertices and $h$ holes. 
Furthermore, let $t \geq 3$ be an integer parameter, to be determined 
later. 
Following Yao~\cite{Yao82} and Clarkson~\cite{Clarkson87}, we 
subdivide the visibility polygon of each vertex in $P$ into 
$t$ \emph{cones} with a small enough apex angle. This 
will allow us to construct compact routing tables that support
a routing algorithm with small stretch factor.

\begin{figure}[htbp]
	\centering
	\includegraphics{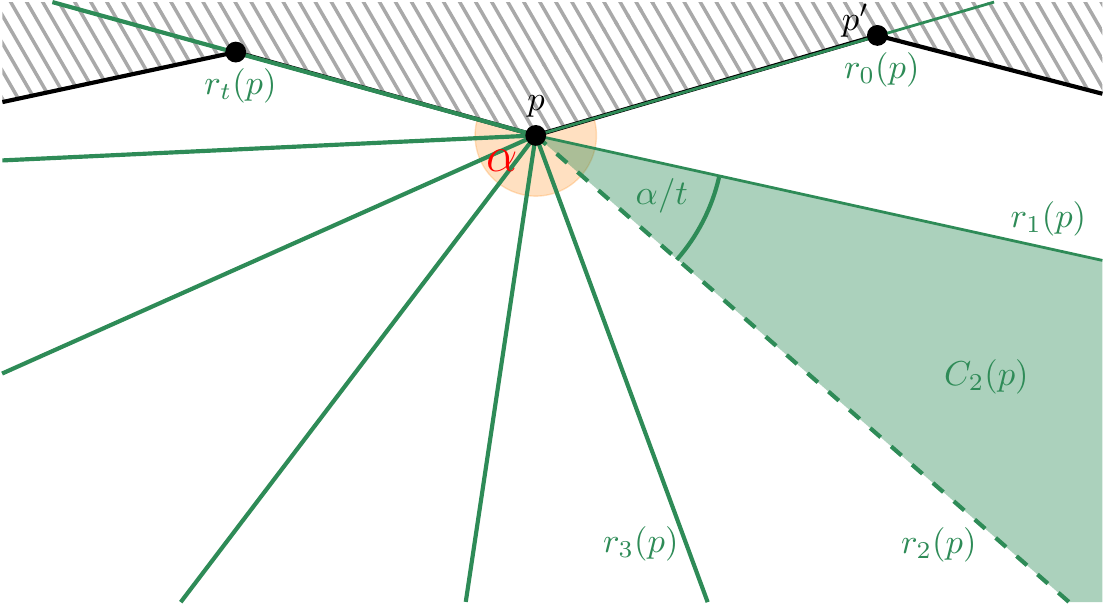}
	\caption{The cones and rays of a vertex $p$ with apex 
      angle $\alpha$.}
	\label{fig:cones}
\end{figure}

Let $p$ be a vertex in $P$ and $p'$ the clockwise 
neighbor of $p$ if $p$ is on the outer boundary, or the 
counterclockwise neighbor of $p$ if $p$ lies on a hole boundary. We
denote with 
$\mathbf{r}(p)$ the \emph{ray} from $p$ through $p'$. To obtain
our cones, we rotate $\mathbf{r}(p)$ by certain angles. Let 
$\alpha$ be the inner angle at $p$. For 
$j = 0, \dots, t$, we write
$r_j(p)$ for the ray $\mathbf{r}(p)$ rotated clockwise by angle 
  $j\cdot\alpha / t$.

Now, for $j = 1, \dots, t$, the cone $C_j(p)$ 
has apex $p$, boundary $r_{j - 1}(p) \cup r_j(p)$, and opening
angle $\alpha/t$; see Figure~\ref{fig:cones}. For 
technical reasons, we define $r_j(p)$ not to be part of $C_j(p)$, for
$1 \leq j < t$, whereas we consider $r_t(p)$
to be part of $C_t(p)$. 
Furthermore, we write $\C(p) = \{ C_j(p) \mid 1 \leq j \leq t \}$ 
for the set of all cones with apex $p$. 
Since the opening angle of each cone is 
$\alpha/t \leq 2\pi/t$ and since $t \geq 3$,  each cone is convex.

The following proof is similar to the one given by 
Clarkson~\cite{Clarkson87} and Narasimhan and 
Smid~\cite{NarasimhanSm07}, though the former shows only that the 
construction leads to an $O(1/\eps)$-spanner instead of showing a 
more precise bound in terms of the number of cones. 

\begin{lemma}
\label{lem:Yao}
Let $p$ be a vertex of $P$ and let $\{p,q\}$ be an edge of 
$\VG(P)$ that lies in the cone $C_j(p)$. Furthermore, let $s$ be 
a vertex of $P$ that lies in $C_j(p)$, is visible from $p$,
and that is closest to $p$. Then, 
$
  d(s,q) \leq \abs{pq} - \left(1 - 2\sin(\pi/t) \right) \abs{ps}.
$
\end{lemma}

\begin{figure}[htbp]
	\centering
	\includegraphics{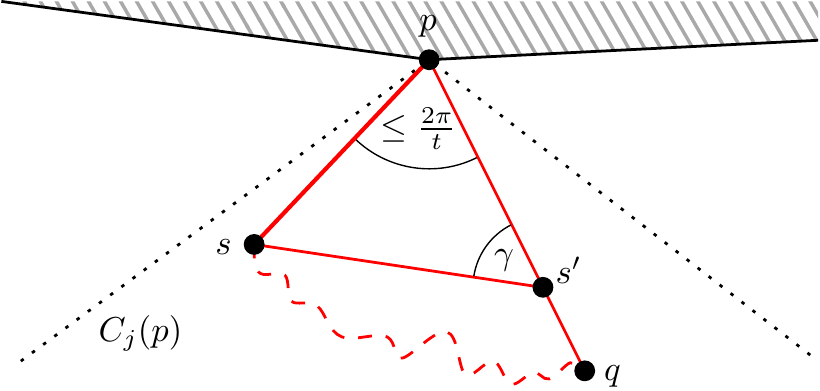}
	\caption{Illustration of Lemma~\ref{lem:Yao}. The points $s$ 
      and $s'$ have the same distance to $p$. The dashed line 
      represents the shortest path from $s$ to $q$.}
	\label{fig:Yao}
\end{figure}

\begin{proof}
Let $s'$ be the point on the line segment $\overline{pq}$ with
$\abs{ps'} = \abs{ps}$; see Figure~\ref{fig:Yao}. Since $p$ can see 
$q$, we have that $p$ can see $s'$ and $s'$ can see $q$. 
Furthermore, $s$ can see $s'$, because $p$ can see $s$ and $s'$
and we chose $s$ to be closest to $p$, so the triangle
$\Delta(p,s,s')$ cannot contain any vertices or (parts of) edges of $P$ in its
interior.  Now, the triangle inequality yields 
$d(s,q) \leq \abs{ss'} + \abs{s'q}$. Let $\beta$ be the inner 
angle at $p$ between the line segments $\overline{ps}$ and 
$\overline{ps'}$. Since both segments lie in the cone $C_j(p)$,
we get $\beta \leq 2\pi / t$. Thus, the angle between
$\overline{s'p}$ and $\overline{s's}$ is $\gamma=\pi/2-\beta/2$. 
Using the sine law and $\sin 2x=2\sin x\cos x$, we get
\[
  \abs{ss'} = \abs{ps} \cdot \frac{\sin\beta}{\sin\gamma}
  = \abs{ps} \cdot \frac{\sin\beta}
     {\sin \left((\pi/2) - (\beta/2)\right)}
  = \abs{ps} \cdot \frac{2\sin (\beta/2)\cos (\beta/2)}
  {\cos (\beta/2)}\leq2\abs{ps}\sin(\pi/t).
\]
Furthermore, we have $\abs{s'q} = \abs{pq} - \abs{ps'}
= \abs{pq} - \abs{ps}$. Thus, the triangle inequality gives
\[
  d(s,q) \leq 2\abs{ps}\sin(\pi/t)+ \abs{pq} - \abs{ps}
  = \abs{pq} - \left(1 - 2\sin(\pi/t)\right)\abs{ps},
\]
as claimed.
\end{proof}

\section{The Routing Scheme}
\label{sec:routinschemePolygons}

Let $\eps > 0$, and let $P$ be a polygonal domain with 
$n$ vertices and $h$ holes. We describe a routing scheme 
for $\VG(P)$ with stretch factor $1 + \eps$. The idea is
to compute for each vertex $p$ the corresponding set of 
cones $\C(p)$ and to store a certain interval of indices for 
each cone $C_j(p)$ in the routing table of $p$.  If an interval 
of a cone $C_j(p)$ contains the target vertex $t$, we proceed to 
the nearest neighbor of $p$ in $C_j(p)$; see 
Figure~\ref{fig:ideaPolygons}. We will see that this results 
in a routing path with small stretch factor.

\begin{figure}[htbp]
	\centering
	\includegraphics{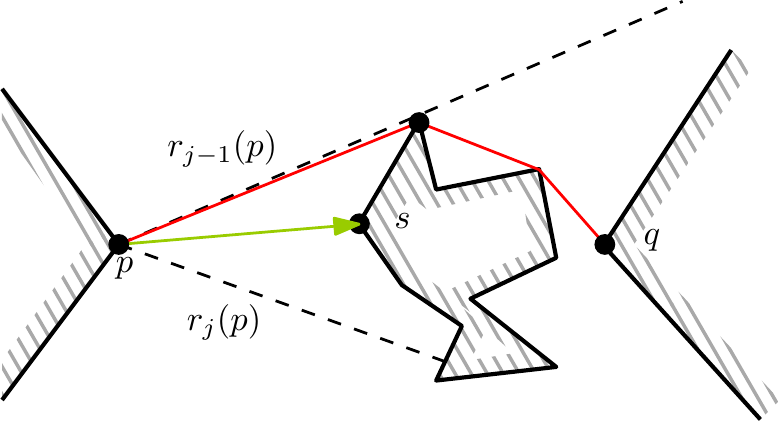}
	\caption{The idea of the routing scheme. The first edge on a 
      shortest path from $p$ to $q$ (red) is contained in $C_j(p)$. 
      The routing algorithm will route the packet from $p$ to 
      $s$ (green), the closest vertex to $p$ in $C_j$.}
	\label{fig:ideaPolygons}
\end{figure}

In the preprocessing phase, we first compute the label of each 
vertex $p_{i, k}$. The label of $p_{i, k}$ is the binary 
representation of $i$, concatenated with the binary representation 
of $k$. Thus, all labels are 
distinct binary strings of length 
$\lceil \log h \rceil + \lceil \log n\rceil$.

Let $p$ be a vertex in $P$. Throughout this section, we 
will write $\C$ and $C_j$ instead of $\C(p)$ and $C_j(p)$. 
The routing table of $p$ is constructed as follows: first, 
we compute a shortest path tree $T$ for $p$. For a vertex
$s$ of $P$, let $T_s$ be the subtree of $T$ with root $s$, 
and denote the set of all vertices on the $i$-th hole in 
$T_s$ by $I_s(i)$. The following well-known observation lies at the 
heart of our routing scheme. 
For completeness, we include a proof.

\begin{observation}
\label{obs:noCrossing}
Let $q_1$ and $q_2$ be two vertices of $P$.
Let $\pi_1$ be the shortest path in $T$ from $p$ to 
$q_1$, and $\pi_2$ the shortest path in $T$ from 
$p$ to $q_2$. Let $l$ be the lowest common ancestor of
$q_1$ and $q_2$ in $T$. Then, $\pi_1$ and $\pi_2$ do not 
cross or touch in a point $x$ with 
$d(p, x) > d(p, l)$.
\end{observation}

\begin{proof}
Suppose first that $\pi_1$ touches $\pi_2$ in a point $x$ with 
$d(p, x) > d(p, l)$. The edges of $T$ are line segments, 
so this can only happen if $x$ is a vertex. But then $T$ would 
contain a cycle, which is impossible.

Next, suppose that $\pi_1$ and $\pi_2$ cross in a point $x$
with $d(p, x) > d(p, l)$. Suppose further that $x$ lies
on the edge $e_1 = (s_1, t_1)$ of $\pi_1$ and the edge 
$e_2 = (s_2, t_2)$ of $\pi_2$; see Figure~\ref{fig:TreesDontCross1}. 
Without loss of generality, we have 
$d(l, s_1) + \abs{s_1 x} \leq d(l, s_2) + \abs{s_2 x}$. 
Since $x \in \innt P$, there is a $\delta > 0$ such that the 
disk $D$ with center $x$ and radius $\delta$ is contained in $P$. 
Now consider the intersection $y_1$ of $\partial D$ with 
$\overline{s_1x}$ and the intersection $y_2$ of $\partial D$ with 
$\overline{xt_2}$.  We have $\overline{y_1y_2} \subset D \subset P$,
and the triangle inequality yields 
$\abs{y_1 x} + \abs{x y_2} > \abs{y_1y_2}$. Hence, the path 
$s_1y_1y_2t_2$ is a shortcut from $l$ to $t_2$, a 
contradiction to $\pi_2$ being a shortest path.
\end{proof}
\begin{figure}[htbp]
	\centering
	\includegraphics[page=2]{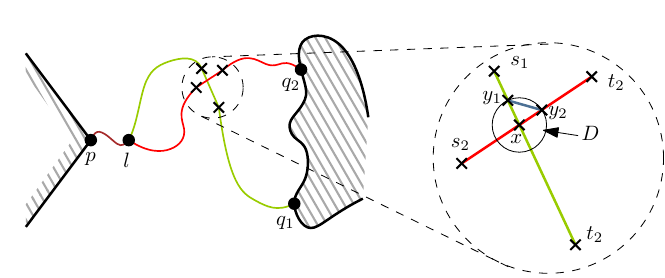}
	\caption{Two shortest paths that originate in $p$ 
	cannot cross.}
	\label{fig:TreesDontCross1}
\end{figure}

\begin{figure}[htbp]
	\centering
	\includegraphics[page=3]{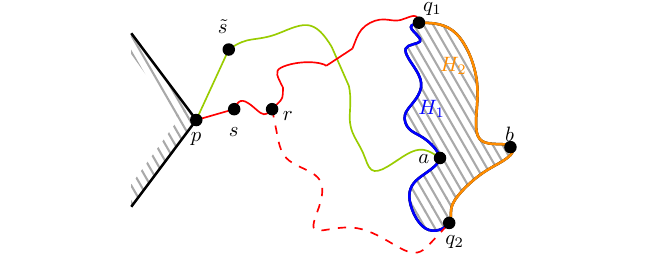}
	\caption{The shortest path from $p$ to $a$ (green) crosses the 
	shortest path from $p$ to $q_1$ (red). This gives a contradiction 
	by Observation~\ref{obs:noCrossing}.}
	\label{fig:TreesDontCross}
\end{figure}

\begin{lemma}
\label{lem:intervals}
Let $e = (p, s)$ be an edge in $T$. Then, the indices of the 
vertices in $I_{s}(i)$ form an interval. Furthermore, let 
$f = (p, s')$ be another edge in $T$, such that $e$ and $f$ are 
consecutive edges in $T$ around $p$.\footnote{By this,
we mean that
there is no other edge of $T$ incident to $p$ in the cone
that is spanned by $e$ and $f$ and that extends into the interior of $P$.}
Then, the 
indices of the vertices in $I_{s}(i) \cup I_{s'}(i)$ are again 
an interval.
\end{lemma}

\begin{proof}
For the first part of the lemma, suppose that the indices 
for $I_{s}(i)$ do not form an interval. Then, there are two 
vertices $q_1, q_2 \in I_{s}(i)$ such that if we consider the two 
polygonal chains $H_1$ and $H_2$ with endpoints
$q_1$ and $q_2$ that constitute the boundary of hole $i$, 
there are two vertices $a, b \notin I_s(i)$ 
with $a \in H_1$ 
and $b \in H_2$ (see Figure~\ref{fig:TreesDontCross}). 
Let $\pi_1$ and $\pi_2$ be the shortest paths in $T$
from $s$ to $q_1$ and from $s$ to $q_2$. 
Let $r$ be the last common
vertex of $\pi_1$ and $\pi_2$, and let $\tilde \pi_1$ be
the subpath of $\pi_1$ from $r$ to $q_1$ and $\tilde \pi_2$
the subpath of $\pi_2$ from $r$ to $q_2$. 
Consider the set $\mathcal{D}$ of 
(open) connected components of 
$P \setminus (\tilde \pi_1 \cup \tilde \pi_2)$. 
Any vertex of $P$ that is on the boundary of
two different components of $\mathcal{D}$ must lie on 
$\tilde \pi_1 \cup \tilde \pi _2$. 
Hence, $p$, $a$, and, $b$ each lie on the boundary of
exactly one component in $\mathcal{D}$, and the components
$D_a$ and $D_b$ with $a$ and $b$ on the boundary are distinct.
Suppose without loss of generality that $p \not \in\partial D_a$. 
Then, there has to be a child $\tilde{s}$ of 
$p$ in $T$ such that
$a \in I_{\tilde{s}}(i)$ and such that the shortest path from 
$\tilde{s}$ to $a$ crosses $\pi_1 \cup \pi_2$. 
Since $p$ is the lowest common ancestor of $a$ and $q_1$ and 
of $a$ and $q_2$, this contradicts Observation~\ref{obs:noCrossing}.

The proof for the second part is very similar.
We assume for the sake of contradiction that the indices
in $I_{s}(i) \cup I_{s'}(i)$ do not form an interval, and we find 
vertices $q_1, q_2 \in I_{s}(i) \cup I_{s'}(i)$ such that if we 
split the boundary of hole $i$ into two chains $H_1$ and $H_2$ 
between $q_1$ and $q_2$, there are two vertices
$a, b \notin I_s(i) \cup I_{s'}(i)$ with $a \in H_1$ 
and $b \in H_2$. Furthermore, we may assume that 
$a \neq p$ and $b \neq p$, because otherwise $q_1$ and $q_2$ would be the
two vertices of $P$ that share an edge with $p$, and thus $q_1$ and $q_2$
would be the only two children of $p$ in $T$ and $I_s(i) \cup I_{s'}(i)$
would be an interval.
Let $\pi_1$ be the shortest path 
in $T$ from $s$ to $q_1$ and $\pi_2$ the shortest path in $T$ from 
$s'$ to $q_2$, and consider the lowest common ancestor $r$ of 
$q_1$ and $q_2$ in $T$ (now $r$ might be $p$). 
Let $\tilde \pi_1$ be
the subpath of $\pi_1$ from $r$ to $q_1$ and $\tilde \pi_2$
the subpath of $\pi_2$ from $r$ to $q_2$. 
Consider the set $\mathcal{D}$ of 
(open) connected components of 
$P \setminus (\tilde \pi_1 \cup \tilde \pi_2)$.
As before, any vertex that lies on the boundaries of
two distinct components of $\mathcal{D}$ must belong to
$\tilde \pi_1 \cup \tilde \pi_2$, so 
$a$ and $b$ are on the boundaries of two 
uniquely defined distinct components in $\mathcal{D}$. We call these
components $D_a$ and $D_b$.
Now, $s$ and $s'$ are consecutive around $p$, so at least one
of $D_a$ and $D_b$ contains no other child of $p$
in $T$ on its boundary. Let it be $D_a$.
Then, the shortest path from $p$ to $a$ must
cross $\pi_1 \cup \pi_2$, contradicting
Observation~\ref{obs:noCrossing}.
\end{proof}

Lemma~\ref{lem:intervals} indicates how to construct the routing 
table $\rho(p)$ for $p$. We set
\begin{equation}
\label{equ:defT}
  t = \left\lceil\pi / \arcsin \left(\frac{1}{2\left(1 + 1/\eps\right)}\right)\right\rceil,
\end{equation}
and we construct a set $\C$ of cones for $p$ as in 
Section~\ref{sec:cones}.  Let $C_j \in \C$ be a cone, and let 
$\Pi_i$ be a hole boundary or the outer boundary. We define 
$C_j \sqcap \Pi_i$ as the set of all vertices $q$ on $\Pi_i$ for
which the first edge of the shortest path from $p$ to
$q$ lies in $C_j$. By Lemma~\ref{lem:intervals}, the indices
of the vertices in $C_j \sqcap \Pi_i$ form a (possibly empty)
cyclic interval $[k_1, k_2]$. If 
$C_j \sqcap \Pi_i = \emptyset$, we do nothing.
Otherwise, if $C_j \sqcap \Pi_i \neq \emptyset$, there is a vertex 
$r \in C_j$ closest to $p$, and we add the entry
$(i, k_1, k_2,  \ell(r))$ to $\rho(p)$. This entry 
needs $2\cdot\lceil \log h \rceil + 3\cdot \lceil \log n \rceil$ 
bits.

Now, the routing function $f : \rho(V) \times \ell(V) \rightarrow V$ is 
quite 
simple. Given the routing table $\rho(p)$ for the 
current vertex $p$ and a target label
$\ell(q) = (i, k)$, indicating vertex $k$ on hole $i$, 
we search $\rho(p)$  for
an entry $(i, k_1, k_2, \ell(r))$ with $k \in [k_1, k_2]$. 
By construction, this entry is unique.
We return $r$ as the next destination for the packet
(see Figure~\ref{fig:ideaPolygons}).

\section{Analysis}

We analyze the stretch factor of our routing scheme and 
give upper bounds on the size of the routing tables and 
the preprocessing time. Let $\eps > 0$ be fixed, and 
let $1 + \eps$ be the desired stretch factor.
We set $t$ as in (\ref{equ:defT}). First, we bound $t$
in terms of $\eps$. This immediately gives 
$\vert\C(p)\vert \in \Oe(1/\eps)$, for every vertex $p$.

\begin{lemma}
\label{lem:epsilonBound}
We have $t\leq 2\pi\left(1+1/\eps\right)+1.$
\end{lemma}
\begin{proof}
For $x \in (0, 1/2]$, we have $\sin x \leq x$, so for 
$z \in [2, \infty)$, we get that $\sin(1 / z) \leq 1/z$. 
Applying $\arcsin(\cdot)$ on both sides, this gives 
$1 / z \leq \arcsin(1 / z) \Leftrightarrow 1/\arcsin(1 / z) \leq z$. 
We set $z = 2(1 + 1/\eps)$ and multiply by $\pi$
to derive the desired inequality.
\end{proof}

\subsection{The Routing Table}\label{sec:table_size}
Let $p$ be a vertex of $P$. We again write 
$\C$ for $\C(p)$ and $C_j$ instead of $C_j(p)$. To bound 
the size of $\rho(p)$, we need some properties of holes 
with respect to cones. For $i = 0, \dots, h - 1$, we 
write $m(i)$ for the number of cones $C_j \in \C$ with 
$\C_j \sqcap \Pi_i \neq\emptyset$. Then, $\rho(p)$ contains
at most
$
\vert\rho(p)\vert \leq \Oe\left(\sum_{i=0}^{h-1}m(i)\log n\right)
$
bits.
We say that $\Pi_i$ is \emph{stretched for the cone $C_j$} if 
there are indices $0 \leq j_1 < j < j_2 < t$  such that 
$C_{j_1} \sqcap \Pi_i$, $C_j \sqcap \Pi_i$ and 
$C_{j_2}\sqcap\Pi_i$ are non-empty. If $\Pi_i$ is not stretched for 
any cone of $p$, then $m(i) \leq 2$. We prove the following lemma:

\begin{lemma}
\label{lem:stretchedHoles}
For every cone $C_j \in \C$, there is at most one 
boundary $\Pi_i$ that is stretched for $C_j$.
\end{lemma}

\begin{proof}
Let $\Pi_i$ be a hole boundary that is stretched for $C_j$. 
There are indices $j_1 <j <j_2$ and vertices 
$q \in C_{j_1} \sqcap \Pi_i$, $r \in C_j \sqcap \Pi_i$, 
and $s \in C_{j_2} \sqcap \Pi_i$. We subdivide $P$ into
three regions $Q$, $R$ and $S$: the boundary of $Q$ is 
given by the shortest path from $p$ to $r$, the shortest
path from $p$ to $q$, and the part of 
$\Pi_i$ from $r$ to $q$ not containing $s$. Similarly,
the region  $R$ is bounded by the shortest path from
 $p$ to $r$, the shortest path from $p$ to $s$ and the
 part of $\Pi_i$ between $r$ and $s$ 
 that does not contain $q$. Finally, $S$ is the 
 closure of $P \setminus (Q \cup R)$. 
The interiors of $Q$, $R$, and $S$ are pairwise 
disjoint;  see Figure~\ref{fig:stretchedHoles}.

\begin{figure}[htbp]
	\centering
	\includegraphics{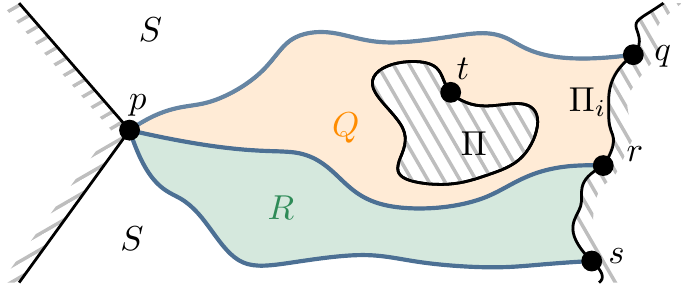}
	\caption{The shortest paths from 
	$p$ to $q$, $r$, $s$ (blue). The hole $\Pi$ contains $t$ 
	and lies in $Q$.}
	\label{fig:stretchedHoles}
\end{figure}

Suppose there is another boundary $\Pi$ that 
is stretched for $C_j$. Then, $\Pi$ must lie 
entirely in either $Q$, $R$, or $S$. We discuss 
the first case, the other two are symmetric. 
Since $\Pi$ is stretched for $C_j$, there is an index
$j' > j$ and a vertex $t \in C_{j'} \sqcap \Pi$. Consider 
the shortest path $\pi$ from $p$ to $t$. Since $j' > j$, the
first edge of $\pi$ lies in $R$ or $S$, and $\pi$ has to
cross or touch the shortest path from $p$ to $q$ or 
from $p$ to $r$. Furthermore, by definition, we 
have $C_j \cap C_{j'} = \{p\}$ and 
$C_{j_1}\cap C_{j'} = \{p\}$. Therefore, $p$ is 
the lowest common ancestor of all three shortest paths, 
and Observation~\ref{obs:noCrossing} leads to a contradiction.
\end{proof}

For $i = 0, \dots, h-1$, let $s(i)$ be the number of cones
in $\C$ for which $\Pi_i$ is stretched. By
Lemma~\ref{lem:stretchedHoles}, we get 
$\sum_{i=0}^{h-1} s(i) \leq \vert\C(p)\vert \in \Oe(1/\eps)$. 
Since $m(i) \leq s(i) + 2$, we conclude
\[
\vert\rho(p)\vert  \in\Oe\left(\sum_{i = 0}^{h -1 }m(i)\log n\right)
=\Oe\left(\sum_{i = 0}^{h - 1}(s(i) + 2)\log n\right)\\
 =\Oe\left((\vert\C(p)\vert + 2h)\log n\right)
=\Oe\left((1/\eps + h)\log n\right).
\]

\subsection{The Stretch Factor}
Next, we bound the stretch factor. First, we prove that 
the distance to the target decreases after the first step. 
This will then give the bound on the overall stretch factor.

\begin{lemma}
\label{lem:polyStretch1}
Let $p$ and $q$ be two vertices in $P$. Let $s$ be the 
next vertex computed by the routing scheme for a
data packet from $p$ to $q$. 
Then, $d(s, q) \leq d(p, q) - \abs{ps}/(1 + \eps)$.
\end{lemma}
\begin{proof}
By construction of $\rho(p)$, we know that the 
next vertex $q'$ on the shortest path from $p$ to $q$ 
lies in the same cone as $s$. Hence, by the triangle
inequality and Lemma~\ref{lem:Yao}, we obtain
\begin{align*}
d(s, q) & \leq d(s, q') + d(q', q)
\leq \abs{pq'} - \left(1 - 2\sin(\pi/t)\right)\abs{ps} + d(q', q)\\
& = d(p,q) - \left(1 - 2\sin(\pi/t)\right)\abs{ps}
 \leq d(p,q)- \left(1 - \frac{1}{1 + 1/\eps}\right)\abs{ps} 
\tag*{(definition of $t$)}\\
& = d(p,q) - \abs{ps}/(1 + \eps),
\end{align*}
as desired.
\end{proof}
Lemma~\ref{lem:polyStretch1} immediately shows the 
correctness of the routing scheme: the distance 
to the target $q$ decreases strictly in each step and 
there is a finite number of vertices, so there is 
a $k = k(p, q) \leq n$ so that after $k$ steps,
the packet reaches $q$. Using this, we can now bound 
the stretch factor of the routing scheme.

\begin{lemma}\label{lem:stretch}
Let $p$ and $q$ be two vertices of $P$. Then, 
$d_{\rho}(p, q)\leq(1 + \eps)d(p,q)$.
\end{lemma}

\begin{proof}
Let $\pi = p_0 p_1 \dots p_k$ be the routing path from 
$p = p_0$ to $q = p_k$. By Lemma~\ref{lem:polyStretch1}, 
we have $d(p_{i + 1},q) \leq d(p_i,q) - \abs{p_ip_{i+1}}/(1 + \eps)$. 
Thus, 
\[
d_{\rho}(p, q) = \sum_{i = 0}^{k - 1}\abs{p_ip_{i+1}}
\leq(1 + \eps)\sum_{i = 0}^{k - 1}\left(d(p_i, q) - 
d(p_{i + 1}, q)\right)\\
=(1 + \eps)\left(d(p_0, q) - d(p_k, q)\right)
=(1 + \eps)d(p, q),
\]
as claimed.
\end{proof}

\subsection{The Preprocessing Time}

Finally, we discuss the details of the preprocessing 
algorithm and its time complexity.

\begin{lemma}\label{lem:preprocess}
The preprocessing time for our routing scheme is 
$\Oe(n^2\log n + n/\eps)$ for polygonal domains and
$O(n^2 + n/\eps)$ for simple polygons.
\end{lemma}

\begin{proof}
Let $p$ be a vertex of $P$. We compute the shortest path tree $T$ for $p$. 
In polygonal domains, this takes $O(n\log n)$ time using the
algorithm of Hershberger and Suri~\cite{HershbergerSu99}, 
and in simple polygons, this needs $O(n)$ time, using 
the algorithm of Guibas~\etal~\cite{GuibasHeLeShTa87}.
We perform a circular sweep around $p$ to find
for each cone $C_j \in \mathcal{C}$ the set $X_j$ of 
the children of $p$ in $T$ that lie in $C_j$. 
This requires $O(n + 1/\eps)$ steps.

For each cone $C_j$,
we find the child $r \in X_j$ that is closest to $p$.
We traverse all subtrees of $T$ that are rooted at some
child in $X_j$, and we collect the set $V_j$ of all their vertices.
We group the vertices in $V_j$ according to the hole boundaries they
belong to. This takes $O(|V_j|)$ time, using the following \emph{bucketing
scheme}: once for the whole algorithm,
we set up an array $B$ of buckets with $h$ entries, one for each
hole boundary. Each bucket consists of a linked list, initially empty. 
This gives a one-time initialization cost of $O(h)$.
When processing the vertices of $V_j$, we create a linked list
$N$ of \emph{non-empty} buckets, also initially empty. 
For each $v \in V_j$, we add $v$ into its corresponding
bucket $B[i]$. If $v$ is the first vertex in $B[i]$, we
add $i$ to $N$. This takes $O(|V_j|)$ time
in total, and it leads to the desired grouping of $V_j$.
Once we have processed $V_j$, we use $N$
in order to reset all the buckets we used to empty, in 
another $O(|V_j|)$ steps.

Now, for each hole $i$, let $V_{j,i}$ be the set of all vertices
on $\Pi_i$ that lie in $V_j$. By 
Lemma~\ref{lem:intervals}, $V_{j,i}$ is a cyclic interval.
To determine its endpoints, it suffices
to identify one vertex on hole $i$ that is not
in $V_{j,i}$ (if it exists). After that, a simple scan over $V_{j,i}$
gives the desired interval endpoints in $O(|V_{j,i}|)$ additional time.
To find this vertex in $O(|V_{j,i}|)$ time, we use prune and search:
let $L = \{ p_{i,k} \in V_{j,i} \mid k < \lceil n_i/2 \rceil\}$
and $R = V_{j,i} \setminus L$. 
We determine $|L|$ and $|R|$ by scanning $V_{j,i}$,
and we distinguish three cases. First, if $|L|=\lceil n_i/2 \rceil$ and
$|R|=\lfloor n_i/2 \rfloor$, all vertices of hole $i$ lie
in the $V_j$, and we are done.
Second, if $|L|<\lceil n_i/2 \rceil$ and $|R|<\lfloor n_i/2 \rfloor$,
then at least one of $p_{i,0}$, $p_{i, \lceil n_i/2 \rceil - 1}$,
$p_{i, \lceil n_i/2 \rceil}$, and $p_{i,n_i-1}$ is not in $V_{j,i}$.
Another scan over $V_{j,i}$ reveals which one it is.
In the third case, exactly one of the two sets $L$, $R$ contains
all possible vertices, whereas the other one does not. We recurse on the
latter set.  This set contains at most $|V_{j,i}|/2$
elements, so the overall running time for the recursion is 
$O(|V_{j,i}|)$.

It follows that we can handle a single cone $C_j$ in time $O(|V_j|)$,
so the total time for processing $p$ is $O(n\log n + 1/\eps)$ in 
polygonal domains and $O(n + 1/\eps)$ in simple polygons.
Since we repeat for each vertex of $P$, the claim follows.
\end{proof}

Combining the last two lemmas with 
Section~\ref{sec:routinschemePolygons}, we get our main theorem.

\begingroup
\def\thetheorem{\ref{mainThm}}
\begin{theorem}
Let $P$ be a polygonal domain with $n$ vertices and $h$ holes. 
For any $\eps > 0$ we can construct a routing scheme for 
$\VG(P)$ with labels of $\Oe(\log n)$ bits and routing tables of 
$\Oe((1/\eps + h)\log n)$ bits per vertex. For any two sites 
$p, q\in P$, the scheme produces a routing path with stretch 
factor at most $1 + \eps$. The preprocessing time is 
$\Oe(n^2\log n)$. If $P$ is a simple polygon, the preprocessing
time reduces to $O(n^2)$.
\end{theorem}
\addtocounter{theorem}{-1}
\endgroup

\begin{proof}
First, note that we may assume that $\eps = \Omega(1/n)$, otherwise,
the theorem follows trivially from storing a complete shortest
path tree in each routing table.
Thus, $1/\eps = O(n)$, and by Lemma~\ref{lem:preprocess}, the 
preprocessing time is
$O(n^2\log n)$ for polygonal domains, and $O(n^2)$ for simple polygons.
The claim on the label size follows from the discussion at the beginning
of Section~\ref{sec:routinschemePolygons}, the size of the routing tables
is given in Section~\ref{sec:table_size}, and the stretch factor is proved
in Lemma~\ref{lem:stretch}.
\end{proof}


\section{Conclusion}
\label{sec:discuss}

We gave an efficient routing scheme for the visibility 
graph of a polygonal domain. Our scheme produces routing 
paths whose length can be made arbitrarily close to the optimum.

Several open questions remain. 
First of all, we would like to obtain
an efficient routing scheme
for the \emph{hop-distance} in polygonal domains $P$, 
where each edge of $\VG(P)$ has unit weight. 
This scenario occurs for routing in a wireless 
network: here, the main overhead is caused by forwarding a packet
at a base station, whereas the distance that the packet has to 
cross is negligible for the travel time.
For our routing scheme, 
we can construct examples where
the stretch factor is $\Omega(n)$; see Figure~\ref{fig:hopExample}. 
Moreover, it would be interesting to improve the preprocessing 
time or the size of the routing tables,
perhaps using a recursive strategy.
\begin{figure}[htbp]
  \centering
  \includegraphics{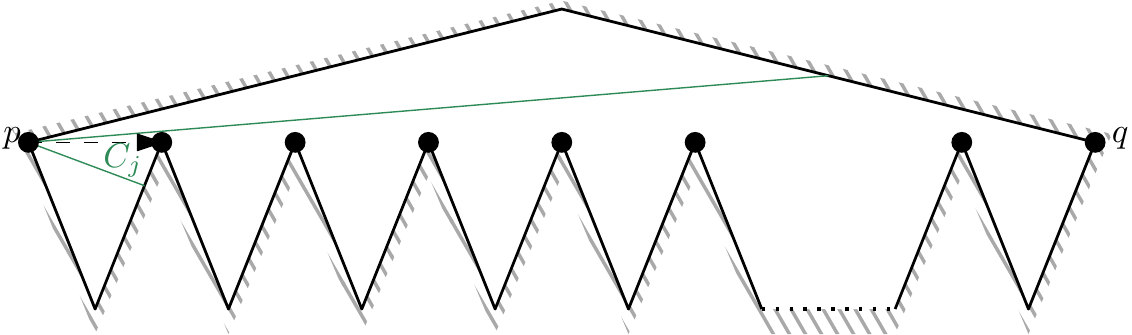}
  \caption{In this polygon, $p$ and $q$ can see each other, so
     their hop-distance is $1$. Our routing scheme 
     routes from one spire to the next, giving stretch factor $\Theta(n)$.}
   \label{fig:hopExample}
\end{figure}

A final open question concerns routing schemes in general: 
how do we model the time needed by a data packet to travel through the 
graph, including the processing times at the vertices? 
In particular, it would be interesting to consider a model
in which each vertex has a fixed \emph{processing time}
until it knows the next vertex for the current packet. This would lead
to a sightly different, but important, measure for routing schemes.

\bibliographystyle{abbrv}
\bibliography{sources}

\end{document}